\documentclass[twoside,leqno,12pt]{article}
\usepackage{amsmath}
\usepackage{amsthm}
\usepackage{amssymb}
\usepackage{amscd}

\usepackage{latexsym}
\usepackage{amsfonts}

\advance\topmargin by -\headsep

\usepackage{newlfont}
\usepackage[dvips]{graphicx}
\usepackage{float}

\def\fu{{\mathfrak{u}}}
\def\cD{{\mathcal{D}}}
\def\cH{{\mathcal{H}}}
\def\cP{{\mathcal{P}}}
\def\cU{{\mathcal{U}}}
\def\tr{{\mathrm{Tr}}}
\def\Pos{\mathrm{Pos}}
\def\P{\mathfrak{P}}
\newcommand{\bra}[1]{\ensuremath{\langle #1 |}}
\newcommand{\ket}[1]{\ensuremath{|\, #1 \rangle}}
\newcommand{\bk}[2]{\ensuremath{\langle #1 | #2 \rangle}}
\newcommand{\kb}[2]{\ensuremath{| #1 \rangle\!\langle #2 |}}

\newtheorem{theo}{Theorem}
\newtheorem{conjecture}{Conjecture}

\newtheorem{lem}{Lemma}
\newtheorem{cor}{Corollary}
\theoremstyle{definition}

\newtheorem{re}{Remark}
\newtheorem{define}{Definition}
\newcommand{\be}{\begin{equation}}
\newcommand{\ee}{\end{equation}}
\newcommand{\ra}{\rightarrow}

\newcommand{\bea}{\begin{eqnarray}}
\newcommand{\eea}{\end{eqnarray}}
\newcommand{\beas}{\begin{eqnarray*}}
\newcommand{\eeas}{\end{eqnarray*}}

\newcommand{\R}{\mathbb{R}}
\newcommand{\C}{\mathbb{C}}

\newcommand{\wh}{\widehat}
\newcommand{\wt}{\widetilde}
\newcommand{\nn}{\nonumber}

\newcommand{\ol}{\overline}

\newcommand{\pa}{\partial}
\newcommand{\ti}{\times}

\newcommand{\Li}{{\cal L}}

\newcommand{\cJ}{{\cal J}}

\newcommand{\Ll}{\Li}

\def\la{\langle}
\def\ran{\rangle}

\def\v{\mathbf{v}}
\def\E{\mathbb{E}}
\def\B{\mathbf{B}}

\def\Aff{\mathfrak{A}}

\mathchardef\za="710B  %\alpha
\mathchardef\zb="710C  %\beta
\mathchardef\zg="710D  %\gamma
\mathchardef\zd="710E  %\delta
\mathchardef\zve="710F %\epsilon
\mathchardef\zz="7110  %\zeta
\mathchardef\zh="7111  %\eta
\mathchardef\zvy="7112 %\theta
\mathchardef\zi="7113  %\iota
\mathchardef\zk="7114  %\kappa
\mathchardef\zl="7115  %\lambda
\mathchardef\zm="7116  %\mu
\mathchardef\zn="7117  %\nu
\mathchardef\zx="7118  %\xi
\mathchardef\zp="7119  %\pi
\mathchardef\zr="711A  %\rho
\mathchardef\zs="711B  %\sigma
\mathchardef\zt="711C  %\tau
\mathchardef\zu="711D  %\upsilon
\mathchardef\zvf="711E %\phi
\mathchardef\zq="711F  %\chi
\mathchardef\zc="7120  %\psi
\mathchardef\zw="7121  %\omega
\mathchardef\ze="7122  %\varepsilon
\mathchardef\zy="7123  %\vartheta
\mathchardef\zf="7124  %\varomega
\mathchardef\zvr="7125 %\varrho
\mathchardef\zvs="7126 %\varsigma
\mathchardef\zf="7127  %\varphi
\mathchardef\zG="7000  %\Gamma
\mathchardef\zD="7001  %\Delta
\mathchardef\zY="7002  %\Theta
\mathchardef\zL="7003  %\Lambda
\mathchardef\zX="7004  %\Xi
\mathchardef\zP="7005  %\Pi
\mathchardef\zS="7006  %\Sigma
\mathchardef\zU="7007  %\Upsilon
\mathchardef\zF="7008  %\Phi
\mathchardef\zW="700A  %\Omega

\begin{document}
\title{Wigner's Theorem \\ and geometry of extreme positive maps}

\author{Janusz Grabowski\footnote{email: jagrab@impan.pl} \\
\textit{Faculty of Mathematics and Natural Sciences, College of Sciences}\\
\textit{Cardinal Stefan Wyszy\'nski University,}\\
\textit{W\'oycickiego 1/3, 01-938 Warszawa, Poland}\\
%\textit{Polish Academy of Sciences, Institute of Mathematics,} \\
%\textit{\'Sniadeckich 8, P.O. Box 21, 00-956 Warsaw, Poland}          \\
\\
Marek Ku\'s\footnote{email: marek.kus@cft.edu.pl}\\
\textit{Center for Theoretical Physics, Polish Academy of Sciences,} \\
\textit{Aleja Lotnik{\'o}w 32/46, 02-668 Warszawa,
Poland} \\
\\
Giuseppe Marmo\footnote{email: marmo@na.infn.it}              \\
\textit{Dipartimento di Scienze Fisiche, Universit\`{a} ``Federico II'' di Napoli} \\
\textit{and Istituto Nazionale di Fisica Nucleare, Sezione di Napoli,} \\
\textit{Complesso Universitario di Monte Sant Angelo,} \\
\textit{Via Cintia, I-80126 Napoli, Italy} \\
}
%\date{}
\maketitle

\begin{abstract}

We consider transformation maps on the space of states which are
symmetries in the sense of Wigner. By virtue of the convex nature of the
space of states, the set of these maps has a convex structure. We
investigate the possibility of a complete characterization of
extreme maps of this convex body to be able to contribute to the
classification of positive maps. Our study provides a variant of
Wigner's theorem originally proved for ray transformations in
Hilbert spaces.



%\medskip
\end{abstract}

\section{Introduction}

Symmetries play a very important role in physics, as it has been
stressed by Wigner on several occasions
\cite{wigner49,wigner64,houtappel65}. The way symmetries are
realized depends on the theory under consideration and more
specifically, according to Felix Klein, on the corresponding
geometric structure of the carrier space, these are `kinematical'
rather than `dynamical' symmetries. It is well known that any
description of physical systems requires the consideration of
states and observables along with a pairing among them providing a
real number with a computable probability \cite{mackey63}. In the
Schr\"odinger-Dirac description of Quantum Mechanics one
associates a Hilbert space with a quantum system, states are
identified with rays of this space and observables are a derived
concept -- they are identified with self-adjoint operators, while
symmetries are defined to be bijections among rays which preserve
probability transitions.

In the $C^*$-algebraic approach to Quantum Mechanics, originated
from the Heisenberg picture, observables are identified with real
elements of this $C^*$-algebra, while states are a derived concept
-- they are identified with positive normalized functionals on the
space of observables. The space of observables carries the
structure of a Jordan algebra and this was the point of view of
Kadison to define symmetries as Jordan-algebra isomorphisms
\cite{kadison51,kadison65}. $C^*$-algebras are quite convenient to
deal with the description of composite systems, the dual space of
states turns out to have a rather involved geometrical
structure. In particular, to take into account the distinction
between separable and entangled states on the space of states, one
is obliged to give up linear superposition in favor of convex
combinations. This change of perspective introduces highly
nontrivial problems, specific to the `convex setting'. As shown
elsewhere \cite{gkm05}, in finite dimensions the space of states
turns out to be a stratified manifold with faces of various
dimensions.

The aim of this paper is to deal with symmetries as those
transformations on the space of states which are appropriate for
its geometrical structure. In doing this we end up with yet
another variant of the celebrated Wigner's theorem on the
realization of symmetries as unitary or antiunitary
transformations on the Hilbert space. The literature on this
theorem, which is also available on text books
\cite{gottfried04,esposito04} in addition to the famous book by
Wigner \cite{wigner59}, is huge. We limit ourselves to a partial
list trying to give a sampling of the various approaches which
have been taken in the years
\cite{hagedorn59,lomont63,oraifeartaigh63,bargmann64,roberts69,hunziker72,
bracci75,samuel97,cassinelli97}.

The paper is organized in the following way: in Section
\ref{sec:denstates} we give a short geometrical description of the
set of density states in a finite-dimensional Hilbert space.
Density states form a convex body in the space of Hermitian
operators. The set of affine maps which map a convex set $K$ into
itself, called simply positive maps, is also a convex set in the
space of affine maps. Characterization of positive maps, e.g.\ by
identifying the extremal ones, i.e.\ maps which can not be
decomposed into a nontrivial convex combination of other positive
maps, can lead to a useful description of the underlying set $K$.
Finding extreme points of such maps is, however, a difficult task,
even if we know explicitly the extreme points of $K$. In Section
\ref{sec:positivemaps} we discuss and give examples of positive
maps for which their extremality can be established upon analyzing
the number of extreme points in the image. In Section
\ref{sec:wigner} we connect the obtained results to the Wigner's
theorem expressed in terms of positive maps which are bijective on pure
states. Sections \ref{sec:CP} and \ref{sec:extremeCP} are devoted
to completely positive maps. In particular, we show again how the
number of extreme points in their image establishes their form and
extremality. We conclude with Section \ref{sec:examples}
containing illustrative examples of extreme positive and
completely positive maps in low dimensions.

\section{Density states}\label{sec:denstates}

Let $\mathcal{H}$ be a finite-dimensional Hilbert space,
$\mathrm{dim}\cH=n$, and let $gl(\cH)$ be the space of complex
linear operators on $\cH$. The space $gl(\cH)$ is canonically a
Hilbert space itself with the Hermitian product $\langle
A,B\rangle={\tr}(A^\dag\circ B)$. As in \cite{gkm05}, we shall
treat the real linear space of Hermitian operators on $\cH$ as the
dual space, $\fu^\ast(\cH)$, of the Lie algebra (of antihermitian
operators), $\fu(\cH)$, of the unitary group $\cU(\cH)$. We have
the obvious decomposition $gl(\cH)=\fu(\cH)\oplus\fu^\ast(\cH)$
into real subspaces with a natural pairing between $\fu(\cH)$ and
$\fu^\ast(\cH)$ given by
\begin{equation}\label{}
\langle A, T\rangle^* =i\cdot\tr(AT),\quad
(A,T)\in\fu^\ast(\cH)\times\fu(\cH),
\end{equation}
and a scalar product induced on $\fu^\ast(\cH)$ by the Hermitian
product and given by
\begin{equation}\label{sp}
\langle A, B\rangle_\ast =\tr(AB),\quad A,B\in\fu^\ast(\cH)\,.
\end{equation}
We denote with $\Vert\cdot\Vert_*$ the corresponding norm.

The coadjoint action of $\cU(\cH)$ on $\fu^\ast(\cH)$ reads
\begin{equation}\label{}
A\mapsto UAU^\dagger, \quad A\in\fu^\ast(\cH), \quad U\in\cU(\cH).
\end{equation}

We denote by $\cP(\cH)$ the space of positive semi-definite
operators from $gl(\cH)$, i.e.\ of those $\rho\in gl(\cH)$ which
can be written in the form $\rho=T^\dag T$ for a certain $T\in
gl(\cH)$. It is a cone, since it is invariant with respect to the
homotheties by $\lambda$ with $\lambda\ge 0$. The set of density
states $\cD(\cH)$ is distinguished in the cone $\cP(\cH)$ by the
equation $\text{Tr}(\rho)=1$, so we will regard $\cP(\cH)$ and
$\cD(\cH)$ as embedded in $u^*(\cH)$.

The space $\cD(\cH)$ is a convex set in the affine hyperplane
$u^*_1(\cH)$ in $u^*(\cH)$, determined by the equation
$\text{Tr}(A)=1$. The model vector space for $u^*_1(\cH)$ is
therefore canonically identified with the space of Hermitian
operators with trace $0$. The space $\Aff(u^*_1(\cH))$ of affine
maps of $u^*_1(\cH)$ can be canonically identified with the space
of these linear maps $\Phi\in\Ll(u^*(\cH))$ which preserve the
trace.

It is known that the set of extreme points of $\cD(\cH)$ coincides
with the set $\cD^1(\cH)$ of pure states, i.e.\ the set of
one-dimensional orthogonal projectors $\kb{x}{x}$, and that every
element of $\cD(\cH)$ is a convex combination of points from
$\cD^1(\cH)$. The space $\cD^1(\cH)$ of pure states can be
identified with the complex projective space $P\cH\simeq\mathbb{C}
P^{n-1}$ \textit{via} the projection
$$\cH\setminus\{ 0\}\ni
x\mapsto\frac{\kb{x}{x}}{\Vert x\Vert^2} \in\cD^1(\cH),
$$
which identifies the points of the orbits of the
$\mathbb{C}\setminus\{ 0\}$-group action by complex homotheties.

If we choose an orthonormal basis $e_1,\dots,e_n$ in $\cH$, we can
identify $u^*(\cH)$ with the real vector space $u^*(n)$ of
Hermitian $n\ti n$ matrices, $u^*_1(\cH)$ with the affine space of
Hermitian $n\ti n$ matrices with trace 1, $U(\cH)$ with the group
$U(n)$ of unitary matrices, $\cD(\cH)$ with $\cD(n)$ - the convex
body of density $n\ti n$ matrices, etc. Recall that the dimension
of $u^*_1(n)$ is $n^2-1$ and the dimension of $u^*(n)$ is $n^2$.

Almost all above can be repeated in the case when $\cH$ is
infinite-dimensional if we assume that all the operators in
question, i.e.\ operators from $gl(\cH)$ and $u^*(\cH)$ are
Hilbert-Schmidt operators (see \cite{gkm07}). The positive
semi-definite operators then, being of the form $AA^\dag$, are
trace-class (nuclear) operators, so density states are trace-class
operators with trace 1. There are some obvious minor differences
with respect to finite dimensions: for instance, the convex set
$\cD(\cH)$ of density states is the closed convex hull of the set
$\cD^1(\cH)$ of pure states, rather than just the convex hull,
etc.

\section{Positive maps of convex sets}\label{sec:positivemaps}

If $K$ is a convex set in a locally convex topological vector
space $E$, then the set $\Pos(K)$ of those continuous linear maps
$\Phi:E\ra E$ which map $K$ into $K$ is a convex set in the (real)
vector space $\Ll(E)$ of all continuous linear maps from $E$ into
$E$. We will refer to elements of $\Pos(K)$ as to {\it linear
$K$-positive maps}, or simply to {\it linear positive maps}, if
$K$ is determined.

If $K$ is compact, then, due to the Krein-Milman Theorem, it is
the closed convex hull of the set $K^0$ of its {\it extreme
points} (points which are not interior points of intervals
included in $K$), $K=\overline{con}(K^0)$. In this sense, compact
convex sets $K$ are completely determined by their extreme points.

However, it should be made clear from the beginning that the
concepts of convex set, positive map, etc., are taken from the
affine rather than linear algebra and geometry. In an affine space
$\E$, one can subtract points, $x=p-p'$, to get vectors of the
model vector space $E=\v(\E)$, or add a vector to a point,
$p=p'+x$, to get another point, but there is no distinguished
point that serves as the origin. More generally, in affine spaces
we can take {\it affine combinations of points}, i.e. combinations
$\sum_i\lambda_ip_i$ such that $\sum_i\zl_i=1$. If all $\zl_i$ are
non-negative the corresponding affine combination is just a {\it
convex combination}. We say that points $p_0,\dots,p_r\in\E$ are
{\it affinely independent} if none is an affine combination of the
others. This is the same as to say that $p_1-p_0,\dots,p_r-p_0\in
E$ are linearly independent vectors.

Convex sets in our approach will live in affine spaces. In this
sense the Krein-Milman Theorem tells us something about compact
convex sets in affine spaces modeled on locally convex linear
spaces.

One can think that the problem is artificial, since by choosing a
point in an affine space as the origin we end up in the model
vector space. However, choosing a point is an additional
information put into the scheme which changes our setting. The
situation is like in a gauge theory, where we can fix a gauge. But
a fixed gauge has, in general, no physical interpretation, so we
rather try to use gauge-invariant objects.

The second instance of affine space presence is the fact that in
many situations, even when we work in a true linear space, it
makes much more sense to admit that positive maps are affine. Note
that affine maps on an affine hyperspace $\E$ of a linear space
$\wh{\E}$ come exactly from linear maps in $\wh{\E}$ which
preserve $\E$. On the other hand, every affine space (or even
affine bundle) $\E$ can be canonically embedded in a linear space
(vector bundle) $\wh{\E}$ as an affine hyperspace (affine
hyperbundle). We refer to \cite{GGU2004,GGU2007} for the
corresponding theory with interesting applications to
frame-independent formulations of some problems in Analytical
Mechanics.

\begin{define} (a) Let $\E_i$ be a real affine space modeled on a
locally convex topological real vector space $E_i=\v(\E_i)$,
$i=1,2$. We say that a map $\Phi:\E_1\ra\E_2$ is an {\it affine
map} if there is a continuous linear map $v(\Phi):E_1\ra E_2$ such
that for any $p\in \E_1$ and any $x\in E_1$, we have
$\Phi(p+x)=\Phi(p)+v(\Phi)(x)$, where $p\mapsto p+x$ is the
natural action of $E_1$ on $\E_1$. The space of affine maps from
$\E_1$ to $\E_2$ will be denoted by $\Aff(\E_1,\E_2)$. If
$\E_1=\E_2=\E$, for the space of affine maps on $\E$, i.e. for
$\Aff(\E,\E)$ we will write shortly $\Aff(\E)$.

(b) Let $\Aff(\E)$ be the space of all affine maps on $\E$ and let
$K$ be a convex set in $\E$. By {\it positive maps on $K$} (or
simply {\it positive maps} if there is no ambiguity about $K$) we
understand these affine maps $\Phi\in\Aff(\E)$ which map $K$ into
$K$. The set of all positive maps on $K$ will be denoted by
$\P(K)$.

(c) By a {\it convex body} we will understand a compact convex set
$K$ with non-empty interior in a finite-dimensional Euclidean
affine space $\E$.
\end{define}

Note that the set $\cD(\cH)$ of density states for quantum
systems with a finite number of levels is an example of a convex
body, as it is canonically embedded in the Euclidean affine space
$u^*_1(\cH)$ of Hermitian operators with trace 1 -- an affine
hyperspace of $u^*(\cH)$.

It is easy to see that for a compact convex set $K$ in a
finite-dimensional affine space $\E$ the closed convex hull
$\ol{con}(K^0)$ is just the convex hull $con(K^0)$ if only
$K^0\subset K$ is closed, and that the convex set of positive maps
$\P(K)$ is again a compact convex set, this time in $\Aff(\E)$.
Note that $\Aff(\E)$ is canonically an affine space modeled on the
vector space $\Aff(\E,E)$ of affine maps from $\E$ into $E$.
Moreover, if $\E$ is just a vector space, $\E=E$, the space
$\Aff(E)$ is a vector space with a canonical decomposition
$\Aff(E)=\Ll(E)\oplus E$, due to the fact that we can write any
affine map $\Phi:E\ra E$ uniquely in the form
$\Phi(x)=v(\Phi)(x)+x_0$ for some $v(\Phi)\in\Ll(E)$ and $x_0\in
E$.

\subsection{Fix-extreme positive maps}

In general it is not easy to find extreme points $\P(K)^0$ of the
convex set of positive maps $\P(K)$, even if extreme points of the
convex body $K$ are explicitly known. This is exactly the case of
the convex bodies $\P(\cD(\cH))$ of positive maps in Quantum
Mechanics.

On the other hand, extremality of some positive maps can be
established relatively easy in the case of maps with many extreme
points in the image, as each extreme point in the image fixes
partially the map. This is based on the observation that, for
$\Phi\in\P(K)$, if $p_0\in K^0$ is the image $p_0=\Phi(p)$ for
some $p\in K$, then $\Phi_i(p)=p_0$ for any $\Phi_i$ of a
decomposition $\Phi=\sum_i\zl_i\Phi_i$ into a convex combination
of $\Phi_i\in\P(K)$. Indeed, as $p_0=\Phi(p)=\sum_i\zl_i\Phi_i(p)$
is a decomposition of the extreme point $p_0$ into a convex
combination of points $\Phi_i(p)\in K$, then $\Phi_i(p)=p_0$. This
immediately implies the following.

\begin{theo}\label{t0} Let $K$ be a compact convex set in an $n$-dimensional real affine space. If a positive map $\Phi\in\P(K)$
has $n+1$ affinely independent extreme points in the image
$\Phi(K)$ of $K$, then the map $\Phi$ is extreme positive,
$\Phi\in\P(K)^0$.
\end{theo}
\begin{proof} Let $q_i\in K$, $i=1,\dots,n+1$, be such that $p_i=\Phi(q_i)$ are extreme and affinely independent
and assume that we have a decomposition $\Phi=t\Phi_0+(1-t)\Phi_1$
for certain $\Phi_0,\Phi_1\in\P(K)$ and $0<t<1$. According to the
observation preceding the above theorem,
$\Phi_0(q_i)=\Phi_1(q_i)=\Phi(q_i)=p_i$ for all $i+1,\dots,n+1$.
But an affine map from a $n$-dimensional affine space is
completely determined by its values on $n+1$ affinely independent
points, so $\Phi_0=\Phi_1=\Phi$.
\end{proof}

The extreme positive maps $\Phi$ described in the above theorem
(with $n+1$ affinely independent extreme points in the image
$\Phi(K)$) will be called {\it fix-extreme positive maps}.
\begin{cor}\label{c1} For any convex body $K$ a positive map $\Psi$ which has all extreme points in $\Psi(K)$
is extreme positive. In particular, the identity map is always an
extreme positive map.
\end{cor}

\subsection{Example: the closed unit ball in $\mathbb{R}^n$}

\begin{theo}\label{t1}
Fix-extreme positive maps of unit balls in Euclidean vector spaces
are orthogonal transformations.
\end{theo}
The proof of the above theorem will be based on the following
lemma.
\begin{lem}
If the function $F:\R^n\ra\R$,
$$F(x_1,\dots,x_n)=\sum_{i=1}^n\left(\za_ix_i+p_i\right)^2\,,$$
where $\za_i>0$ and $p_i\in\R$, $i=1,\dots,n$, has on the unit
sphere $S^{n-1}=\{x\in\R^n:\sum_ix_i^2=1\}$ local maxima at some
$n+1$ affinely independent points $q_1,\dots,q_{n+1}\in S^{n-1}$,
then $F$ is constant on $S^{n-1}$. In particular,
$p_1=\dots=p_n=0$ and $\za_1=\za_2=\dots=\za_n$.
\end{lem}
\begin{proof}
We will use the method of Lagrange multipliers and consider the
function
$$F_\zl(x_1,\dots,x_n)=\sum_{i=1}^n\left(\za_ix_i+p_i\right)^2-
\zl\left(\sum_{i=1}^nx_i^2\right)\,.$$ Since $q_j$,
$j=1,\dots,n+1$, are critical points of $F$, when restricted to the sphere,
the coordinates $(x_1,\dots,x_n)$ of each $q_j$ solve the system
of equations
\bea\label{e1}\frac{\pa F_\zl}{\pa x^i}(x)&=&(\za_i^2-\zl)x_i+p_i\za_i=0,\ i=1,\dots,n\,,\\
\sum_{i=1}^nx_i^2&=&1\,.\nn \eea Moreover, as at $q_j$ we have
local maxima, the second derivative of $F_\zl$ must be
non-positive definite that yields $\za_i^2-\zl\le 0$ for all $i$.
We can assume that $\za_1\ge\za_2\ge\cdots\ge\za_n$, so
$\zl\ge\za_1^2$.

Assume first that $\zl>\za_1^2$. Then,
\be\label{s}
x_i=\frac{p_i\za_i}{\zl-\za_i^2},
\ee and
$$G(\zl)=\sum_{i=1}^n\frac{(p_i\za_i)^2}{(\zl-\za_i^2)^2}$$
should be equal to 1. But the function $G(\zl)$ is monotone with
respect to $\za_1^2<\zl<+\infty$, so there is at most one solution
$(x,\zl_0)$ of (\ref{e1}) with $\zl_0>\za_1^2$. It must be
therefore at least $n$ additional solutions with $\zl=\za_1^2$.
Let $k$ be the number of the biggest $\za_i$, i.e.
$\za_1=\za_2=\dots=\za_k>\za_{k+1}$. We get easily from (\ref{e1})
that
\be\label{b}
b_1=\dots=b_k=0.
\ee Hence if $k=n$ then $F$ is constantly $\za_1^2$ on the sphere. It suffices to show now that $k<n$ is not possible.
Indeed, if $x$ is a solution of (\ref{e1}) with $\zl=\za_1^2$, we
get
$$x_i=\frac{p_i\za_i}{\za_1^2-\za_i^2}\,,\quad i=k+1,\dots,n\,,$$
so we can have, together with the solution (\ref{s}) corresponding
to $\zl>\za_1^2$, an affinely independent set of $n+1$ solutions
only if $k=n-1$. But then, due to (\ref{b}), the solution
(\ref{s}) must be
\be\label{s0}x=(0,\dots,0,sgn(p_n))\,, \ee so $p_n\ne 0$ and $G(\zl)$ reduces
to
$$G(\zl)=\frac{(p_n\za_n)^2}{(\zl-\za_n^2)^2},
$$
and
\be\label{G}G(\zl_0)=\frac{(p_n\za_n)^2}{(\zl_0-\za_n^2)^2}=1\,.\ee
On the other hand, the solutions corresponding to $\zl=\za_1^2$
must be of the form
$$\left(x_1,\dots,x_{n-1},\frac{p_n\za_n}{\za_1^2-\za_n^2}\right),
$$
so that we have solutions additional to (\ref{s0}) only if
\be\label{G1}\frac{(p_n\za_n)^2}{(\za_1^2-\za_n^2)^2}< 1\,. \ee
But, as $\zl_0>\za_1^2$, the latter contradicts (\ref{G}):
$$1=\frac{(p_n\za_n)^2}{(\zl_0-\za_n^2)^2}<\frac{(p_n\za_n)^2}{(\za_1^2-\za_n^2)^2}< 1\,.$$
\end{proof}
Now we can prove Theorem \ref{t1}.
\begin{proof}
Let $\B$ be the unit ball in an $n$-dimensional Euclidean vector
space $E$. Let us make an identification of $E$ with $\R^n$ with
the standard Euclidean norm
$$\Vert x\Vert^2=\sum_{i=1}^nx_i^2\,.$$
Let $\Phi$ be a fix-extreme positive map of $\B$, i.e.
$\Phi:\R^n\ra\R^n$ is an affine map, $\Phi(x)=A(x)+p'$, where
$A\in\Ll(\R^n)$ is a linear map of $\R^n$ such that
$\Phi(\B)\subset\B$ and $\Phi(\B)$ has $n+1$ affinely independent
points $\Phi(q'_1),\dots,\Phi(q'_{n+1})$ in the unit sphere
$S^{n-1}$. Of course, we can assume that $q'_1,\dots,q'_{n+1}$ are
extreme points of $\B$, so they lie on the sphere as well. In
particular the map (matrix) $A$ is invertible.

Now we can apply the singular value decomposition to the matrix
$A$ in order to write it in the form $A=O_1\circ T\circ O_2$,
where $O_1,O_2$ are orthogonal matrices and $T$ is a diagonal
matrix with positive entries $\za_1,\dots\za_n>0$ on the diagonal.
Since we can write
$$\Phi(x)=A(x)+b'=O_1\circ T\circ O_2(x)+b'=O_1(T\circ O_2(x)+b)\,,$$
where $O_1(b)=b'$, and since the orthogonal maps preserve $\B$ and
$S^{n-1}$, the map $\Phi_0(x)=T(x)+b$ has the same properties as
$\Phi$: it is a positive map of $\B$, $\Phi_0(\B)\subset\B$ and
$\Phi_0(\B)$ has $n+1$ affinely independent points
$\Phi_0(q_1),\dots,\Phi_0(q_{n+1})$ on the unit sphere $S^{n-1}$,
where $q_j=O_2^{-1}(q'_j)$ are points of the sphere,
$j=1,\dots,n+1$. This means that the function
$$F(x)=\Vert\Phi_0(x)\Vert^2=\sum_{i=1}^n\left(\za_ix_i+p_i\right)^2\,,$$
reduced to the unit sphere takes at $q_1,\dots,q_{n+1}$ local
maxima. Applying the above lemma we conclude that $b=0$ and $F$ is
constant on the sphere, so that $\Phi_0=T$ maps the unit sphere
into the unit sphere. Hence, $T=I$ and $\Phi=O_1\circ O_2$ is
orthogonal.
\end{proof}
\begin{re} Theorem \ref{t1} can be derived also from the results of \cite{gorini76}.
\end{re}

\subsection{Example: an extreme map on the plane fixing two extreme points}

In the present section we want to present a simple example of an
extreme map in two dimensions which is a bijection on its two
extreme points. To this end let us consider the function on the
interval $[1,1]$,
\begin{equation}\label{eq:f}
f(x)=\left(\frac{1-x}{2}\right)^{2}\left(\frac{1+x}{2}\right)^{1/2}+
2\left(\frac{1-x}{2}\right)^{1/2}
\left(\frac{1+x}{2}\right)^{5/2}+\frac{1-x^2}{4}\,.
\end{equation}
The function $f$ is concave, hence the subset $S$ of the $(x,y)$
plane bounded by its graph and the interval $[1,1]$ is convex. Let
us perform a linear transformation of the $(x,y)$ plane,
\begin{equation}\label{eq:lin}
T:(x,y)\mapsto\left(-x,\frac{y}{2}\right)\,.
\end{equation}
Under this transformation $S$ is transformed into the set bounded
by $[1,1]$ and the graph of
\begin{equation}\label{eq:g}
g(x)=\frac{1}{2}f(-x),
\end{equation}
i.e.\
\begin{equation}\label{eq:g1}
g(x)=\frac{1}{2}\left(\frac{1+x}{2}\right)^{2}\left(\frac{1-x}{2}\right)^{1/2}
+\left(\frac{1+x}{2}\right)^{1/2}
\left(\frac{1-x}{2}\right)^{5/2}+\frac{1-x^2}{8}\,.
\end{equation}
Since $f(x)\ge g(x)$ for $x\in[1,1]$, we have $T(S)\subset S$.
Moreover $T$ is a bijection on two extremal points, $(x,y)=(-1,0)$
and $(x,y)=(1,0)$, of $S$. Observe also that $T$ is an extreme
mapping in the sense that for an arbitrary $\alpha\ge 1$ there is
$x\in [-1,1]$ such that $f(x)-\alpha g(x)<0$, i.e.\ the linear
transformation
\begin{equation}\label{eq:lin1}
T_\alpha:(x,y)\mapsto\left(-x,\alpha\frac{y}{2}\right)
\end{equation}
does not map $S$ into $S$.

The above described properties of $f$ and $g$ can be established
by straightforward calculations. Below we illustrate them in
Figures 1-3.
\begin{figure}[H]
\centering
\includegraphics[scale=.5,angle=-90]{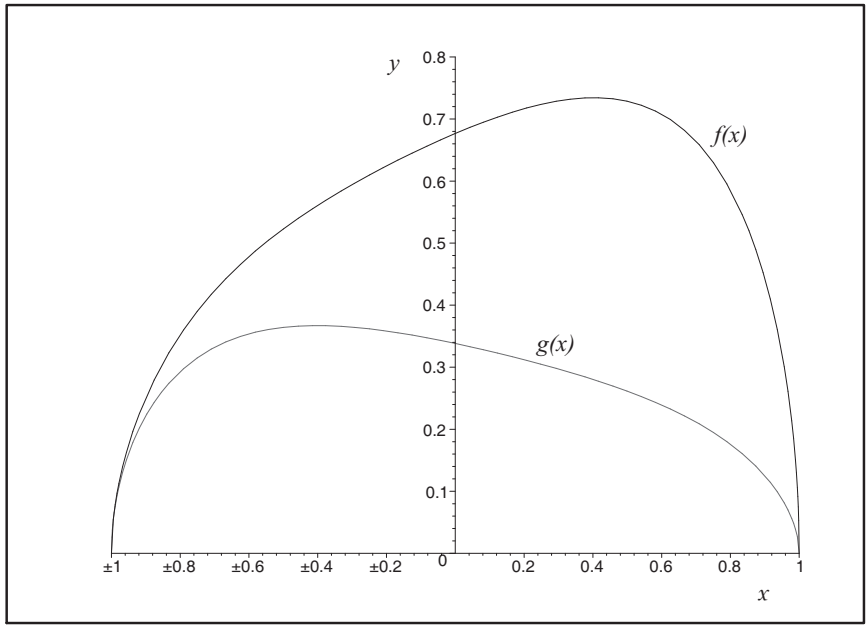}
\caption{Functions $f(x)$ and $g(x)$.} \label{fig:fg}
\end{figure}
\begin{figure}[H]
\centering
\includegraphics[scale=.5,angle=-90]{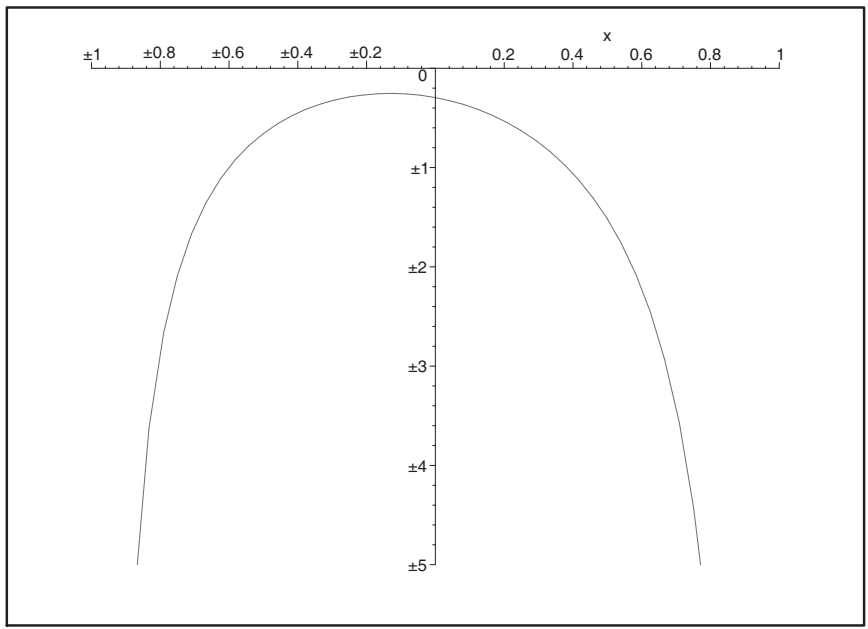}
\caption{Second derivative of $f$.} \label{fig:f2}
\end{figure}
\begin{figure}[H]
\centering
\includegraphics[scale=.5,angle=-90]{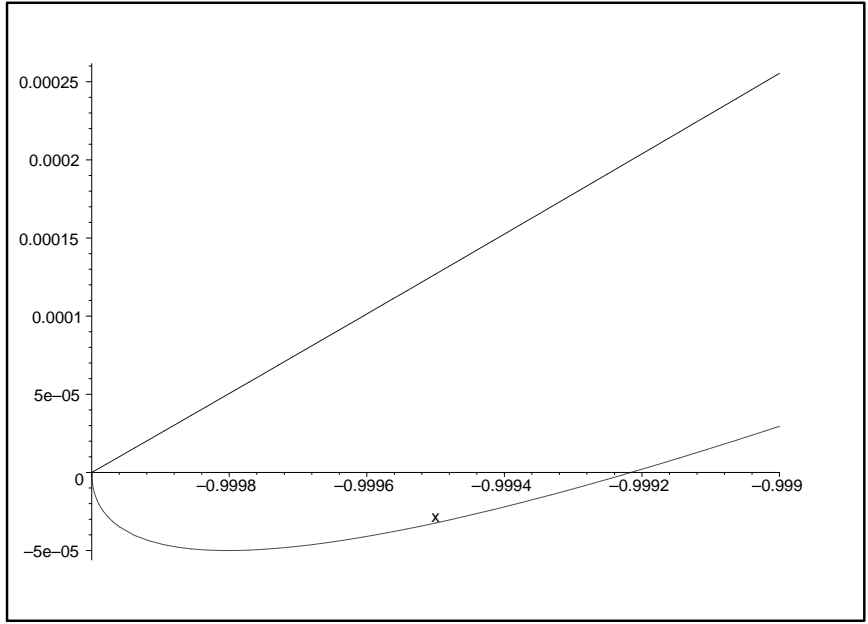}
\caption{Functions $f-g$ (upper curve) and $f(x)-1.01\cdot
g(x)$ (lower curve) in the vicinity of $x=-1$.} \label{fig:fmg}
\end{figure}

\section{Positive maps bijective on pure states - a version of Wigner's
Theorem}\label{sec:wigner}

Before we formulate a version of the Wigner's Theorem
\cite{bargmann64} let us comment on complex antilinear and
antiunitary maps in a Hilbert space. A map $A:\cH\ra\cH$ we call
{\it antilinear} if $A(\za x+\zb y)=\bar{\za}x+\bar{\zb}y$, where
$\bar{\za}$ denotes the complex conjugation of $\za\in\C$. An
antilinear map $U:\cH\ra\cH$ we call {\it antiunitary} if
$\bk{Ux}{Uy}=\bk{y}{x}$ for all $x,y\in\cH$. The adjoint $A^\dag$
of an antilinear map is an antilinear map defined {\it via} the
identity
$$\bk{Ax}{y}=\bk{A^\dag y}{x}\,.$$
Any linear (antilinear) map $A:\cH\ra\cH$ induces a linear (resp.,
antilinear) map $M_A:gl(\cH)\ra gl(\cH)$ which on one-dimensional
maps $\kb{x}{y}$ takes the form $M_A(\kb{x}{y})=\kb{Ax}{Ay}$. For
linear $A$ we can easily represent the map $M_A$ as $M_A(\zr)=A\zr
A^\dag$, while with antilinear maps the situation is a little bit
more complicated.

If an orthonormal basis is chosen then in the Hilbert space $\cH$
we can define a complex conjugation
$$C:\cH\ra\cH\,,\quad C^2=I\,,\quad C\left(\sum_ia_ie_i\right)=\sum_i\bar{a_i}e_i\,.$$
Instead of $Cx$ we will write simply $\bar{x}$. It is clear that
$\bk{x}{y}=\bk{\bar{y}}{\bar{x}}$. If $A$ is a complex linear map
then $\wt{A}=A\circ C$ is antilinear and {\it vice versa}. Since
any continuous complex linear (antilinear) map $A:\cH\ra\cH$
is represented by a (possibly infinite) matrix $(a_{ij})$, where
$A(e_i)=\sum_ja^j_ie_j$, also the transposition $A\mapsto A^T$ is
well defined:
$$A^T(e_i)=\sum_ja^i_je_j,$$
and we extend it to the whole $\cH$ by complex linearity
(antilinearity). For linear $A$ the adjoint map $A^\dag$ can be
then written as
$$A^\dag=C\circ A^T\circ C\,,$$
so that $C\circ A\circ C=A^T$ for linear Hermitian $A=A^\dag$. If
$A$ is antilinear then $\wt{A}=A\circ C$ is linear, so
$$\kb{Ax}{Ay}=\kb{\wt{A}\bar{x}}{\wt{A}\bar{y}}=\wt{A}\circ \kb{\bar{x}}{\bar{y}}\circ(\wt{A})^\dag\,.$$
But, as easily seen,
$$\kb{\bar{x}}{\bar{y}}=C\circ\kb{x}{y}\circ C\,,$$
so that, for Hermitian $\zr$, \be\label{al}
M_A(\zr)=\wt{A}\zr^T(\wt{A})^\dag\,. \ee

The Wigner's Theorem (compare with \cite{bargmann64}) can be now
formulated as follows.
\begin{theo} Let $\psi:\cD^1(\cH)\ra\cD^1(\cH)$ be a bijection of pure states in a Hilbert space $\cH$
preserving the transition probabilities \be\label{tp}
\bk{\psi(\zr_1)}{\psi(\zr_2)}_*=\bk{\zr_1}{\zr_2}_*\,. \ee Then,
there is a unitary $U:\cH\ra\cH$ such that
\be\label{wt}\psi(\zr)=U\zr U^\dag\quad\text{for all pure states }\zr\,, \ee or
\be\label{wt1}\psi(\zr)=U\zr^T U^\dag\quad\text{for all pure states }\zr\,, \ee where $\zr\mapsto\zr^T$ is the
transposition associated with a choice of an orthonormal basis in
$\cH$.
\end{theo}
The standard versions of Wigner's Theorem usually consider (unit)
vectors of the Hilbert space rather than pure states. But if $x,y$
are unit vectors representing pure states $\zr_1$ and $\zr_2$,
respectively, then
$$\la\zr_1,\zr_2\ran_*=|\la x,y\ran|^2\,,$$
so that preserving $|\la x,y\ran|$ is the same as preserving
$\la\zr_1,\zr_2\ran_*$. Moreover, any unitary (or antiunitary)
action  in the Hilbert space, $x\mapsto Ux$, induces on pure
states $\zr=\kb{x}{x}$ the action (\ref{wt}) or (\ref{wt1}). We will call the
maps (\ref{wt}) and (\ref{wt1}) defined on pure states or on
$u^*(\cH)$ {\it Wigner maps}. The Wigner maps on
$u^*(\cH)$ can be abstractly characterized as follows.
\begin{theo} A linear map $\psi:u^*(\cH)\ra u^*(\cH)$ is a Wigner map if and only if it is positive and
orthogonal.
\end{theo}
\begin{proof} The Wigner maps are clearly positive and orthogonal, so let us assume that $\psi$
has these properties. For all Hermitian $\zr_1,\zr_2$ we have
therefore (\ref{tp}) and we know that $\psi(\zr)$ is positive
semi-definite if $\zr$ is. The map $\psi$ is orthogonal, therefore
invertible, and its inverse $\psi^{-1}$ is orthogonal as well. Let
us observe that $\psi^{-1}$ is also a positive map. Take a pure
state $\zr$ and suppose that $\psi^{-1}(\zr)$ has the
spectral decomposition $\psi^{-1}(\zr)=\zr_+-\zr_-$ into a
difference of positive semi-definite operators $\zr_+,\zr_-$ which
are orthogonal, $\bk{\zr_+}{\zr_-}_*=0$. Then $\zr$ is a
difference of orthogonal positive semi-definite operators
$\zr=\psi(\zr_+)-\psi(\zr_-)$ and, as $\zr$ is a pure state,
$\psi(\zr_-)$ (thus $\zr_-$) must be 0. A similar argument shows
that the image $\psi(\zr)$ of any pure state $\zr$ is a positive
semi-definite operator which is not decomposable into a sum of
orthogonal positive semi-definite operators, so $\psi(\zr)$ is a
pure state up to a constant factor. Since
$$\tr(\psi(\zr)^2)=\bk{\psi(\zr)}{\psi(\zr)}=\bk{\zr}{\zr}=1\,,$$
this factor equals 1 and we conclude that $\psi$ induces a
bijection on pure states.
\end{proof}

We will now prove a theorem which extends Wigner's Theorem and
which relates it to the problem of extreme positive maps.

Let $u^*_f(\cH)$ be the linear subspace of $u^*(\cH)$ consisting
of Hermitian finite-rank operators. For $K_1,K_2\subset
u^*_1(\cH)$ we say that a map $\psi:K_1\ra K_2$ is {\it affine} if
$\psi$ is the restriction to $K_1$ of a trace-preserving linear
map $\Phi:<K_1>\ra <K_2>$ from the linear span $<K_1>$ of $K_1$ in
$u^*(\cH)$ into the linear span $<K_2>$ of $K_2$ in $u^*(\cH)$,
$\Phi(K_1)\subset K_2$.

\begin{theo}Let $\psi:\cD^1(\cH)\ra\cD^1(\cH)$ be a bijective map. The following are equivalent:
\begin{itemize}
\item[(a)]  $\psi$ is affine; \item[(b)] $\psi$ preserves
transition
    probabilities between pure states;
\item[(c)] $\psi$ is a Wigner map.
\end{itemize}
If any (or all) of these cases is satisfied, there is a unique
continuous affine extension $\Psi:u^*_1(\cH)\ra u^*_1(\cH)$ of
$\psi$ which is extreme positive, $\Psi\in\P(\cD(\cH))^0$.
\end{theo}
\begin{proof}
$(a)\Rightarrow(b)$. Since $u^*_f(\cH)$ is spanned by the set
$\cD^1(\cH)$ of pure states, let $\Phi:u^*_f(\cH)\ra u^*_f(\cH)$
be the (unique) linear trace-preserving map on the space
$u^*_f(\cH)$ of finite-rank Hermitian operators inducing $\psi$ on
$\cD^1(\cH)$. Since $\Phi$ maps convex combinations into convex
combinations, $\Phi$ maps finite-rank density states into
finite-rank density states, so $\Phi$ is a positive map. We will
prove that $\Phi$ is a linear isomorphism on $u^*_f(\cH)$. This
follows from the fact that $\Phi$ preserves the rank, i.e. induces
a bijection on $\cD^k(\cH)$ for each $k=1,2,\dots$.

To see the latter let us remark that the rank, $rank(\zr)$, of
$\zr\in u^*_f(\cH)$ is defined as a minimal number of pure states
whose linear combination is $\zr$, so that, as $\Phi$ is linear
and is a bijection on pure states, $rank(\Phi(\zr))\le rank(\zr)$.
Conversely, if $\zr\in\cD(\cH)$ is of rank $k$, then it is a
convex combination of some pure states $\zr_1,\dots,\zr_k$, thus
the image by $\Phi$ of a convex combination of pure states
$\Phi^{-1}(\zr_1),\dots,\Phi^{-1}(\zr_k)$. This shows that
$\cD^k(\cH)\subset\Phi(\cD^k(\cH))$. The linear map $\Phi$ induces
a bijection on $\cD^1(\cH)$, so assume inductively that it induces
a bijection on $\cD^l(\cH)$ for $l\le k$. Let now $\zr\in
u^*_f(\cH)$ be of rank $k+1$, $\zr\in \cD^{k+1}(\cH)$, with the
spectral decomposition $\zr=\sum_{i=0}^k\zl_i\zr_i$, with
$\zl_i>0$, $\sum_i\zl_i=1$, and $\zr_i=\kb{x_i}{x_i}$ being
pairwise orthogonal pure states, $\bk{x_i}{x_j}=0$, $i\ne j$. We
must show that the rank of $\Phi(\zr)$ is $k+1$.

Suppose the contrary. Hence, according to the inductive
assumption, $\zr'=\Phi(\zr)$ is of rank $k$. As
$\zr=\zl_0\zr_0+\wt{\zr}$, where
$\wt{\zr}=\sum_{i=1}^k\zl_i\zr_i$, thus $\wt{\zr}'=\Phi(\wt{\zr})$
is of rank $k$, the image of $\Phi(\zr_0)=\kb{x'_0}{x'_0}$, thus
$x'_0$ belongs to the image of $\wt{\zr}'$. Consider the spectral
decomposition $\wt{\zr}'=\sum_{i=1}^k\zl'_i\zr'_i$, with
$\zl'_i>0$, $\sum_i\zl'_i=1$, and $\zr'_i=\kb{x'_i}{x'_i}$ being
pairwise orthogonal pure states, $\bk{x'_i}{x'_j}=0$, $i\ne j$.
Let $\cH_0$ (resp., $\cH'_0$) be the subspace in $\cH$ spanned by
the vectors $x_1,\dots,x_k$ (resp., $x'_1,\dots,x'_k$) and let
$\cD^1(\cH_0)$ (resp., $\cD^1(\cH'_0)$) be the set of all pure
states of $\cH_0$ (resp., $\cH'_0$), i.e. these pure states from
$\cD^1(\cH)$ which are represented by unit vectors from $\cH_0$
(resp., $\cH'_0$). It is now clear that
$\zr'_0=\Phi(\zr_0)\in\cD^1(\cH'_0)$. Note that pure states $\zh$
from $\cD^1(\cH'_0)$ can be characterized as such pure states
which added to $\wt{\zr}'$ do not change the rank,
$rank(\wt{\zr}'+\zh)=rank(\wt{\zr}')=k$. According to the
inductive assumption this implies that
$rank(\wt{\zr}+\Phi^{-1}(\zh))=k$ as well, but
$$rank(\wt{\zr}+\Phi^{-1}(\zr'_0))=rank(\wt{\zr}+\zr_0)=rank(\zr)=k+1\,,$$
which is a contradiction.

Since we know now that $\Phi$ induces bijections on each
$\cD^k(\cH)$, $k=1,2,\dots$, it is easy to conclude that it is a
rank-preserving isomorphism, so that $\Phi^{-1}$ is also a
positive map. Indeed, as $\cD^1(\cH)$ spans $u^*_f(\cH)$, it is
clearly "onto". It is also injective, since
$\Phi(\zl\zr-\zl'\zr')=0$, where $\zl,\zl'>0$ and
$\zr,\zr'\in\cD(\cH)$ are density states of finite ranks, implies
($\Phi$ is positive) that $\Phi(\zr)=\Phi(\zr')=0$, thus
$\zr=\zr'=0$, as $\Phi$ preserves the rank of density states.

To finish the proof, we will need the following lemma.
\begin{lem} Let $\zr\in\cD(\cH)$ be a density state of a finite rank $k$.
Then the square of the Hilbert-Schmidt norm
$\Vert\zr\Vert_*^2=\tr(\zr^2)$ can be characterized as the maximum
of the expressions $\sum_{i=1}^k\zl_i^2$ over all decompositions
$\zr=\sum_{i=1}^k\zl_i\zr_i$ of $\zr$ as a convex combination of
$k$ pure states $\zr_1,\zr_2,\dots,\zr_k$. This maximum is
associated with the spectral decomposition.
\end{lem}
\begin{proof} As
\be\label{hs1}\tr(\zr^2)=\sum_{i}\zl_i^2+2\sum_{i\ne
j}\zl_i\zl_j\tr(\zr_i\zr_j), \ee and \be\label{hs0}\sum_{i\ne
j}\zl_i\zl_j\tr(\zr_i\zr_j)\ge 0\,, \ee we have
\be\label{hs}\Vert\zr\Vert_*^2\ge\sum_{i=1}^k\zl_i^2\,. \ee
Moreover, we have equality in (\ref{hs}) if and only if we have
equality in (\ref{hs0}), so as all $\zl_i>0$, if and only if
$\bk{\zr_i}{\zr_j}_*=\tr(\zr_i\zr_j)=0$ for all $i\ne j$. This
corresponds to the spectral decomposition.
\end{proof}

The above lemma implies that the map $\Phi$ preserves the
Hilbert-Schmidt norm of density states. Indeed, if we use the
spectral decomposition to write $\zr$ as a convex combination
$\zr=\sum_{i=1}^k\zl_i\zr_i$ of pure states, then $\Phi(\zr)$ can
be expressed as a convex combination of pure states
$\Phi(\zr)=\sum_{i=1}^k\zl_i\Phi(\zr_i)$ with the same
coefficients, so that
$\Vert\Phi(\zr)\Vert_*^2\ge\sum_i\zl_i^2=\Vert\zr\Vert_*^2$. But
we can apply the above consideration to $\Phi^{-1}$ instead of
$\Phi$ and get $\Vert\Phi^{-1}(\zh)\Vert_*^2\ge\Vert\zh\Vert_*^2$
for any density state of rank $k$, in particular for
$\zh=\Phi(\zr)$. We get therefore
$\Vert\Phi(\zr)\Vert_*^2=\Vert\zr\Vert_*^2$.

Let us now take two pure states $\zr_1,\zr_2$ and consider
$\zr=\frac{1}{2}(\zr_1+\zr_2)$. Since, according to (\ref{hs1}),
$$\Vert\zr\Vert_*^2=\frac{1}{2}\left(1+\bk{\zr_1}{\zr_2}_*\right),$$
and
$$\Vert\zr\Vert_*^2=\Vert\Phi(\zr)\Vert_*^2=\Vert\frac{1}{2}
\left(\Phi(\zr_1)+\Phi(\zr_1)\right)\Vert_*^2\,,$$
 we have
$$\frac{1}{2}\left(1+\bk{\zr_1}{\zr_2}_*\right)=\frac{1}{2}
\left(1+\bk{\Phi(\zr_1)}{\Phi(\zr_2)}_*\right)\,,$$ thus
$$\bk{\zr_1}{\zr_2}_*=\bk{\Phi(\zr_1)}{\Phi(\zr_2)}_*=\bk{\psi(\zr_1)}{\psi(\zr_2)}_*\,,$$
so $\psi$ preserves transition probabilities between pure states.

$(b)\Rightarrow(c)$ is the Wigner's Theorem.

$(c)\Rightarrow(a)$ is obvious.

Moreover, $\psi$ has an obvious unique continuous extension
$\Psi:u^*_1(\cH)\ra u^*_1(\cH)$, $\Psi(\zr)=U\zr U^\dag$ or
$\Psi(\zr)=U\zr^T U^\dag$. Since $\Psi$ is positive and has all
extreme points in its image, it is extreme positive according to
the obvious infinite-dimensional version of Corollary \ref{c1}.
\end{proof}

If the dimension of the Hilbert space $\cH$ is $n$, then the
dimension of the affine space $u^*_1(\cH)$ of Hermitian operators
with trace 1 equals $n^2-1$ and we know from the general theory
(Theorem \ref{t0}) that a positive map $\Phi\in\P(\cD(\cH))$
possessing $n^2$ pure states in the image $\Phi(\cD(\cH))$ of
$\cD(\cH)$ is extreme positive (we called such positive maps
fix-extreme). We finish this section with the following conjecture
motivated by Theorem \ref{t1}.
\begin{conjecture} Any fix-extreme positive map $\Phi\in\P(\cD(\cH))^0$ is a Wigner map.
\end{conjecture}

\section{Completely positive maps acting on pure states}\label{sec:CP}

A linear map $A: u^\ast(\mathcal{H})\rightarrow
u^\ast(\mathcal{H})$ is called {\it completely positive} (CP) if
$A\otimes I_N: u^\ast(\mathcal{H})\otimes M_N\rightarrow
u^\ast(\mathcal{H})\otimes M_N$, where $M_N$ is the algebra of
complex $N\times N$ matrices, is positivity-preserving for all
$N$. It was shown by Choi \cite{choi75} and Kraus \cite{kraus83}
that each CP map admits a representation in the so called Kraus
form
\begin{equation}\label{cp}
A\rho=\sum_{i=1}^sV_i^\dagger\rho V_i,
\end{equation}
where $V_k$ are operators acting on $\mathcal{H}$.

\begin{define}We say that a CP operator (\ref{cp})
acting on Hermitian operators of a Hilbert space $\mathcal{H}$ is
{\it extreme} if any decomposition, $A=A_1+\cdots+A_r$, into a sum
of CP operators $A_1,\ldots,A_r$ is {\it irrelevant}, i.e. all the
operators $A_1,\ldots,A_r$ are proportional to $A$, $A_k=a_kA$,
for some $a_k\in\mathbb{R}_+$. In particular, all operators $V_i$
are proportional and $A$ can be written as a single Kraus operator
\[
A\rho=V^\dagger\rho V.
\]
\end{define}

Not to deal with the cone of CP operators but rather with a
compact convex set (if the dimension of $\cH$ is finite), one has
to put certain normalization condition. We can try to use, for
instance, the trace $\tr(A)$ of a CP operator (\ref{cp}) as a
linear operator on the complex vector space $gl(\cH)$ which is the
same as the trace of (\ref{cp}) as an operator on the real vector
space $u^*(\cH)$ of Hermitian operators.
\begin{theo} If $A$ is a CP operator in the form (\ref{cp}) then
\be\label{tr}\tr(A)=\sum_{i=1}^s|\tr(V_i)|^2 \,. \ee
\end{theo}
\begin{proof} It suffices to prove (\ref{tr}) for a single Kraus map $A=M_V$.
Choose an orthonormal basis $\{e_i\}$ in $\cH$ and write $V$ in
this basis as a complex matrix $V=(v_{ij})$. As an orthonormal
basis in $gl(\cH)$ we can take $\zr_{jk}=\kb{e_j}{e_k}$. We have
\beas\tr(M_V)&=&\sum_{j,k}\bk{\zr_{jk}}{V^\dag\zr_{jk}V}=
\sum_{j,k}\tr\left(\zr_{kj}V^\dag\zr_{jk}V\right)\\
&=&
\sum_{j,k}\bk{e_k}{\zr_{kj}V^\dag\zr_{jk}Ve_k}=\sum_{j,k}\bar{v}_{jj}v_{kk}=|\tr(V)|^2\,.
\eeas
\end{proof}

As we can see, the trace can vanish for non-zero CP operators
which makes the normalization by trace impossible. There is,
however, another possibility of normalizing CP operators provided
by the {\it Jamio\l\-kowski isomorphism} \cite{jamiolkowski72}
\be\label{j} \cJ:gl(gl(\cH))\ra gl(gl(\cH))\,. \ee The Jamio\l
kowski isomorphism maps the operator $M_A^B(\zr)=A\zr B^\dag$ into
the rank-one operator $\kb{A}{B}$. In particular, CP operators
correspond, {\it via} the Jamio\l kowski isomorphism $\cJ$, to
positive semi-definite operators on $gl(\cH)$ (see, for instance,
\cite{gkm07}). Any Kraus operator $M_V(\rho)=V\rho V^\dag$
corresponds to the one dimensional Hermitian operator $\kb{V}{V}$.
The spectral decomposition of $\cJ(A)$ results in the
decomposition (\ref{cp}) with $V_i$ being mutually orthogonal with
respect to the Hilbert-Schmidt product $\tr(A^\dag B)$ in
$gl(\cH)$. Such a decomposition of a CP operator we will call a
{\it spectral decomposition}. We will call a CP operator {\it
normalized} if it corresponds, {\it via} the Jamio\l kowski
isomorphism, to a trace-one operator. If (\ref{cp}) is a spectral
decomposition of a CP operator, then it is normalized if
$$\tr\left(\sum_{i=1}^s V_i^\dag V_i\right)=1\,.$$
We will call $\tr_s(A)=\tr\left(\sum_{i=1}^s V_i^\dag V_i\right)$
the {\it spectral trace} of a CP operator $A$. The spectral trace
is greater or equal to 0 and equals 0 only if $A=0$. The convex
set of normalized CP operators on $gl(\cH)$ we will denote by
$NCP(\cH)$. It is clear that extreme points in $NCP(\cH)$ are
exactly single normalized Kraus maps.

From now on we assume the Hilbert space $\mathcal{H}$ to be of a
finite dimension $n$.

Recall that on the space of Hermitian operators on $\mathcal{H}$
we have a canonical scalar product
$\langle\rho,\rho^\prime\rangle_*=\tr\rho\rho^\prime$ and that
$\cP(\cH)$ denotes the cone of positive semi-definite operators.
We will call elements $\rho_1,\ldots,\rho_k$, $k\ge n$, of
$\cP(\cH)$ to be {\it in general position} if for any non-zero
$\rho\in \cP(\cH)$ we cannot find $n$ of them which are orthogonal
to $\rho$. If $\rho_i$ are of rank-one, $\rho_i=\kb{x_i}{x_i}$,
where $\ket{x_i}\in\mathcal{H}$, this simply means that any $n$
vectors of $\ket{x_1},\ldots,\ket{x_k}$ form a linear basis in
$\mathcal{H}$.
\begin{theo}\label{theo:1}
The following statements are equivalent:
\begin{description}
\item{(a)} $A$ is invertible and $A^{-1}$ is a CP operator.
\item{(b)}
    $A$ is invertible and extreme, i.e.
\[
A\rho=V^\dagger\rho V
\]
with $V$ - invertible. \item{(c)} For any set of pure states
$\rho_1,\ldots,\rho_{n+1}$ in general position,
$A\rho_1,\ldots,A\rho_{n+1}$ are operators of rank one in general
position. \item{(d)} In the image $A(P(\cH))$ there are $n+1$
rank-one operators in
    general position.

\item{(e)} There exists a set of pure states
$\rho_1,\ldots,\rho_{n+1}$ in
    general position such that $A\rho_1,\ldots,A\rho_{n+1}$ are operators
    of rank one in general position.
\end{description}
\end{theo}
\begin{proof}

(a)$\Rightarrow$(b). It was proven in \cite{grabowski06} (Theorem
7).

(b)$\Rightarrow$(c). If $\rho=\kb{x}{x}$ then $V^\dagger\rho
V=\kb{V^\dagger x}{V^\dagger x}$, hence $A$ maps pure states into
rank-one operators, and $V$, being invertible, preserves all
linear independencies.

(c)$\Rightarrow$(d) - trivial.

(d)$\Rightarrow$(e). Let $\rho_1,\ldots,\rho_{n+1}\in P(\cH)$ be
positive semi-definite operators such that $\eta_j=A\rho_j$,
$j=1,\ldots,n+1$, are rank-one operators in general position.
First, let us remark that we can assume that $\rho_j$ are density
states, since proportionality plays no role here. Second, we can
assume further that they are pure. Indeed, if $\rho$ is a state
with the spectral decomposition $\rho=\sum_k a_k\xi_k$ into a
convex combination of pure states $\xi_k$ and if $A\rho=\eta$ is
rank-one positive semi-definite operator, then $\eta=A\rho=\sum_k
a_kA\xi_k$ is a convex combination of rank-one positive
semi-definite operators $A\xi_k$, so $\eta$ is positively
proportional to $A\xi_k$ for all $k$, since pure
states are extreme points in the convex body of all states. By a
similar argument, all states $V_i^\dagger\rho_j V_i$,
$i=1,\ldots,s$, are proportional to $\eta_j$, say
$V_i^\dagger\rho_j V_i=\beta_j\eta_j$, where, of course,
$\sum_{i=1}^s\beta_i^j=1$.

Let us write $\rho_j=\kb{x_j}{x_j}$, $\eta_j=\kb{y_j}{y_j}$,
$j=1,\ldots,n+1$, for some vectors $\ket{x_j}$, $\ket{y_j}$. We
claim that the states $\rho_1,\ldots,\rho_{n+1}$ are in general
position, i.e. any $n$ among the vectors
$\ket{x_1},\ldots,\ket{x_{n+1}}$ are linearly independent.

For, assume the contrary, i.e.\ that, say
$\ket{x_1},\ldots,\ket{x_n}$, are linearly dependent. Hence, one
of the vectors, say $\ket{x_n}$, can be written as a linear
combination $\ket{x_n}=b_1\ket{x_1}+\cdots b_{n-1}\ket{x_{n-1}}$.
Since $V_i^\dagger\rho_j V_i=\ket{V_i^\dagger x_j}\bra{V_i^\dagger
x_j}$ is proportional to $\eta_j=\kb{y_j}{y_j}$, the vector
$\ket{V_i^\dagger x_j}$ is proportional to $\ket{y_j}$ for
$j=1,\ldots,n+1$. In particular, all the vectors $\ket{V_i^\dagger
x_j}$ with $i=1,\ldots,s$ and $j=1,\ldots,n-1$, belong to the
linear span $span\langle \ket{y_1},\ldots,\ket{y_{n-1}}\rangle$ of
the vectors $\ket{y_1},\ldots,\ket{y_{n-1}}$. Therefore
\[
V_i^\dagger\ket{x_n}=b_1V_i^\dagger\ket{x_1}+\cdots+b_{n-1}V_i^\dagger\ket{x_{n-1}}\in
span\langle \ket{y_1},\ldots,\ket{y_{n-1}}\rangle.
\]
Since $V_i^\dagger\ket{x_n}$ is proportional to $\ket{y_n}$, and
$\ket{y_1},\ldots,\ket{y_n}$ are linearly independent, we have
$V_i^\dagger\ket{x_n}=0$ for all $i=1,\ldots,s$. But this
contradicts the property $A\rho_n=\eta_n$, i.e.
$\sum_{i=1}^sV_i^\dagger\ket{x_n}=\ket{y_n}$.

(e)$\Rightarrow$(a). Put $\rho_j=\kb{x_j}{x_j}$,
$\eta_j=A\rho_j=\kb{y_j}{y_j}$ for some
$\ket{x_j},\ket{y_j}\in\mathcal{H}$, $j=1,\ldots,n+1$. Since
$\eta_j=A\rho_j$ is of rank one, all the positive semi-definite
operators $V_i^\dagger\rho_j V_i$ of its decomposition
\[
\eta_j=\sum_{i=1}^s V_i^\dagger\rho_j V_i,
\]
must be proportional to $\eta_j$, $V_i^\dagger\rho_j
V_i\sim\eta_j$. This, in turn, means that $V_i^\dagger\ket{x_j}$
are proportional to $\ket{y_j}$,
\[
V_i^\dagger\ket{x_j}=\alpha_i^j\ket{y_j}, \quad i=1,\ldots,s,
\quad j=1,\ldots,n+1.
\]
As any $n$ vectors of $\ket{x_1},\ldots,\ket{x_{n+1}}$ are
linearly independent, all the coefficients $a_j$ of the
decomposition
\[
\ket{x_{n+1}}=a_1\ket{x_1}+\cdots +a_n\ket{x_n}
\]
are non-zero, $a_j\ne 0$. Since
\[
V_i^\dagger\ket{x_{n+1}}=a_1V_i^\dagger\ket{x_1}+\cdots
+a_nV_i^\dagger\ket{x_n}= a_1\alpha_i^1\ket{y_1}+\cdots
+a_n\alpha_i^n\ket{y_n}
\]
are proportional to $\ket{y_{n+1}}$ and $\ket{y_{n+1}}$ has a
decomposition $\ket{y_{n+1}}=b_1\ket{y_1}+\cdots +b_n\ket{y_n}$
into a linear combination of linearly independent
$\ket{y_1},\ldots,\ket{y_n}$, the vector
$\left(\alpha_i^1,\ldots,\alpha_i^n\right)\in\mathbb{C}^n$ must be
proportional to
$\left(a_1/b_1,\ldots,a_n/b_n\right)\in\mathbb{C}^n$. In
consequence, it means that the operators $V_i^\dagger$ are
proportional to the operator $V^\dagger$ uniquely defined by the
conditions
\[
V^\dagger\ket{x_j}=\frac{b_j}{a_j}\ket{y_j}.
\]
Hence
\[
A\rho=\beta V^\dagger\rho V
\]
for some $\beta\in\mathbb{R}_+$ and, clearly, $\beta\ne 0$.
\end{proof}
\begin{re}
It is not enough to take only $n$ pure states
$\rho_1,\ldots,\rho_n$. Let us take, for example, $V_i$ diagonal,
$V_i=\mathrm{diag}\left(\lambda_i^1,\ldots,\lambda_i^n\right)$,
$\lambda_j^k\ne 0$, and not proportional. If now
$\lambda_i^1=\lambda_i^2$, we have an infinite set of pure states
that are mapped to rank-one operators and any $n$ of them are not
in general position, but $A=\sum_{i=1}^sV_i^\dagger\rho V_i$ is
not extreme.
\end{re}

\begin{cor}\label{theo:2} If a CP operator
\begin{equation}\label{}
A\rho=\sum_{i=1}^s V_i^\dagger\rho V_i
\end{equation}
is trace-preserving (resp., unity-preserving) and such that the
image of density states $A(\mathcal{D})$ contains $n+1$ pure
states in general position, then $A$ is unitary,
$A\rho=U^\dagger\rho U$.
\end{cor}
\begin{proof}It is enough to make use of Theorem~\ref{theo:1} and observe that trace-preserving (unity-preserving) yields
$VV^\dag=I$ (resp., $V^\dag V=I$), so $V$ is unitary.

\end{proof}

\section{Extreme bistochastic CP maps}\label{sec:extremeCP}

Extreme unity-preserving CP maps have been described by Man-Duen
Choi \cite{choi75}. We can reformulate his result as follows.
\begin{theo}\label{choi1} Let $NCP_I(\cH)$ (resp., $NCP_\tr(\cH)$) be the convex body of normalized unity-preserving
(resp. trace-preserving) CP maps on $gl(\cH)$. Then $A\in
NCP_I(\cH)$ (resp. $A\in NCP_\tr(\cH)$), with the spectral
decomposition
\be\label{bi}A\rho=\sum_{i=1}^s V_i^\dagger\rho V_i \ee is extreme if and
only if the operators $\{ V_i^\dag V_j: i,j=1,\dots,s\}$ (resp.,
$\{ V_i V_j^\dag: i,j=1,\dots,s\}$) are linearly independent in
$gl(\cH)$.
\end{theo}
\begin{proof} We will sketch a proof making use of the Jamio\l kowski isomorphism (\ref{j}) which associates with the CP
map (\ref{bi}) the Hermitian operator
$$\cJ(A)=\sum_{i=1}^s \kb{V_i^\dagger}{V_i^\dagger}$$
on $gl(\cH)$. Extreme normalized CP operators correspond
therefore, {\it via} the Jamio\l kowski isomorphism, to pure
states on the Hilbert space $gl(\cH)$. On the space $u^*(gl(\cH))$
of Hermitian operators acting on the Hilbert space $gl(\cH)$, in
turn, we can define two canonical $\R$-linear maps
$F_1,F_2:u^*(gl(\cH))\ra u^*(\cH)$ which associate with rank-one
operators $\kb{V}{V}$ the operators $F_1(\kb{V}{V})=VV^\dag$ and
$F_2(\kb{V}{V})=V^\dag V$, respectively. The unity-preserving CP
operators $A$ correspond therefore to Hermitian operators
constrained by the equations $F_1(\cJ(A))=I$. The corresponding
extreme points $\cJ(A)\in u^*(gl(\cH))$ need not be pure
states. It suffices that the level sets of $I$ of the linear
constraint $F_1$ are transversal to the face of the point, i.e.
the function $F_1$ has the trivial kernel on the tangent space
$T_{\cJ(A)}$ of the face at $\cJ(A)$. But this tangent space is
known (see e.g. \cite{gkm05,grabowski06}) to consists of all
Hermitian operators with the range equal to the range of $\cJ(A)$,
i.e.\ the operators of the form
$\sum_{i,j}\zl^{ij}\kb{V_i^\dag}{V_j^\dag}$, where $(\zl^{ij})$ is
a Hermitian matrix. Hence $A$ is extremal in $NCP_I(\cH)$ if and
only if
\be\label{choi}F_1(\sum_{i,j}\zl^{ij}\kb{V_i^\dag}{V_j^\dag})=\sum_{i,j}\zl^{ij}{V_i^\dag}{V_j}\ne
0\in u^*(\cH) \ee for all Hermitian $(\zl^{ij})\ne 0$. As we can
decompose any operator into the sum of a Hermitian and an
antihermitian ones, we can rewrite (\ref{choi}) in $gl(\cH)$
instead of $u^*(\cH)$ using an arbitrary complex matrix
$(\zl^{ij})\ne 0$. This is nothing but the complex linear
independence of the operators $\{ V_i^\dag V_j: i,j=1,\dots,s\}$
in $gl(\cH)$. For trace-preserving CP operators the reasoning is
identical with the function $F_2$ replacing $F_1$.
\end{proof}

The maps from $NCP_*(\cH)=NCP_I(\cH)\cap NCP_\tr(\cH)$ are
sometimes called {\it bistochastic}. The above understanding of
the Choi's result gives us easily a characterization of extreme
bistochastic maps.
\begin{theo} A bistochastic map $A\in NCP_*(\cH)$ with
the spectral decomposition
$$A\rho=\sum_{i=1}^s V_i^\dagger\rho V_i$$
is extreme if and only if the operators $\{ V_i^\dag V_j\oplus V_i
V_j^\dag: i,j=1,\dots,s\}$ are linearly independent in
$gl(\cH\oplus\cH)$.
\end{theo}
\begin{proof} The proof is analogous to the above one with the difference that our constraint function is now
$$(F_1,F_2):u^*(gl(\cH))\ra u^*(\cH)\ti u^*(\cH)\simeq u^*(\cH)\oplus u^*(\cH)\,.$$
The condition (\ref{choi}) is therefore replaced by the condition
$$\sum_{i,j}\zl^{ij}{V_i^\dag}{V_j}\ne 0\quad\text{or}\quad \sum_{i,j}\zl^{ij}{V_i}{V^\dag_j}\ne 0
\quad\text{for all Hermitian }(\zl^{ij})\ne 0\,. $$ We can rewrite
it as
$$\sum_{i,j}\zl^{ij}\left({V_i^\dag}{V_j}\oplus{V_i}{V^\dag_j}\right)\ne 0$$
and pass to an arbitrary complex $(\zl^{ij})\ne 0$ like before.
\end{proof}

\section{Examples}\label{sec:examples}

Let us illustrate some of the previous reasonings and results in
the simplest cases of maps on states on two- and three-dimensional
spaces. In the following subsection we show three examples of
extreme maps: a generic one possessing exactly two pure states in
its image, a non-generic one with only one pure state in the
image, and an extreme map having a continuous family of pure
states in the image. The last one, according to
Theorem~\ref{theo:2}, is not completely positive but a merely
positive extreme map. Note that these cases have been considered
also in \cite{ruskai02}.

In the second subsection we give an example of an extreme
completely positive map acting on $\mathbb{C}^{3\times 3}$ which
does not have any pure state in its image. Such a situation is
impossible for maps acting on qubits (i.e. maps on
$\mathbb{C}^{2\times 2}$).

\subsection{Extreme completely positive, positive, stochastic and bistochastic
maps for $n=2$}

A state on $\mathbb{C}^2$ can be parameterized by a unit vector
$(x,y,z)\in\mathbb{R}^3$
\begin{equation}\label{}
\rho=\frac{1}{2}\left(I+x\sigma_1+y\sigma_2+z\sigma_3\right),
\end{equation}
where
\begin{equation}\label{}
\sigma_1=\left[
\begin{array}{cc}
 0 & 1 \\
 1 & 0 \\
\end{array}
\right], \quad \sigma_2=\left[
\begin{array}{cc}
 0 & -i \\
 i & 0 \\
\end{array}
\right], \quad \sigma_3=\left[
\begin{array}{cc}
 1 & 0 \\
 0 & -1 \\
\end{array}
\right].
\end{equation}
A positive trace-preserving map
\begin{equation}\label{2dimmap}
\rho\mapsto
\rho^\prime=\frac{1}{2}\left(I+x^\prime\sigma_1+y^\prime\sigma_2+z^\prime\sigma_3\right)
\end{equation}
is thus determined (up to rotations which are irrelevant from the point of view of this paper) by four parameters
$\lambda_1,\lambda_2,\lambda_3,t$, such that
\begin{equation}\label{}
x^\prime=\lambda_1 x,\quad y^\prime=\lambda_2 y,\quad
z^\prime=\lambda_3 z+t.
\end{equation}
The parameters $\lambda_1,\lambda_2,\lambda_3,t$ must fulfil
particular conditions to ensure the positivity of the map (see
\cite{fujiwara99}, \cite{ruskai02} for details).

The image of the unit sphere $x^2+y^2+z^2=1$ under (\ref{2dimmap})
is the ellipsoid
\begin{equation}\label{ell}
\left(\frac{x}{\lambda_1}\right)^2+\left(\frac{y}{\lambda_2}\right)^2+\left(\frac{z-t}{\lambda_3}
\right)^2=1.
\end{equation}
Obviously, for a positive map (\ref{2dimmap}) the ellipsoid is
inside the unit sphere. For extreme maps it has to have points on
its surface common with the surface of the unit sphere (i.e. some
pure states are mapped into pure states). For extreme CP maps two
possibilities occur \cite{ruskai02}:
\begin{enumerate}
\item $\lambda_1=\cos(u)$, $\lambda_2=\cos(v)$,
    $\lambda_3=\cos(u)\cos(v)$, $t=\sin(u)\sin(v)$, $<0<u<v$. In this
    case the ellipsoid (\ref{ell}) has three different axes and it
    touches the unit sphere at two points (see
    Fig.~\ref{fig:cp-extreme-generic}).

\vskip 2cm

\begin{figure}[H]
\begin{center}
\includegraphics[scale=0.4]{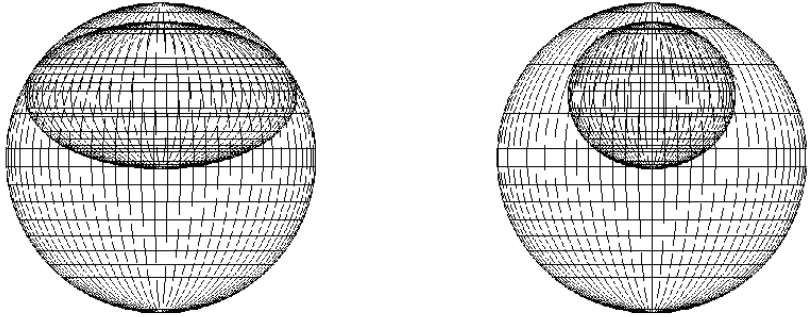}
\caption{An extreme CP map having exactly two pure states in the
image. Projections along two perpendicular axes.}
\label{fig:cp-extreme-generic}
\end{center}
\end{figure}

\item $\lambda_1=\lambda_2=\cos(u)$, $\lambda_3=\cos^2(u)$,
    $t=\sin^2(u)$, in which case the ellipsoid (\ref{ell}) touches the
    unit sphere at a single point $x=y=0$, $z=1$. (see
    Fig.~\ref{fig:cp-extreme-singular}).

\vskip 2cm

\begin{figure}[H]
\begin{center}
\includegraphics[scale=0.4]{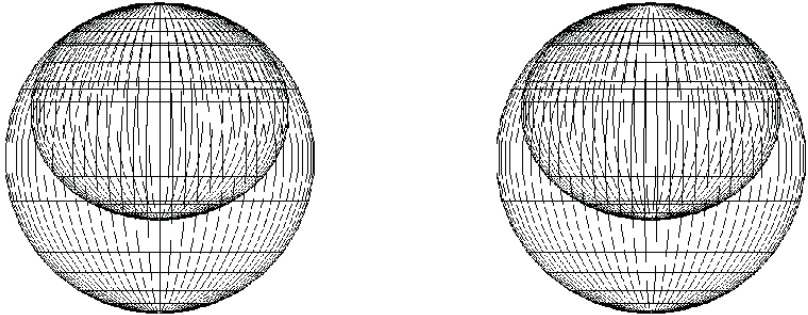}
\caption{An extreme CP map having exactly one pure state in the
image. Projections along two perpendicular axes.}
\label{fig:cp-extreme-singular}
\end{center}
\end{figure}

\item Geometrically it is obvious that without upsetting the
extremality
    of the map we can make the ellipsoid (\ref{ell}) touching the unit
    sphere along a full circle (see Fig.~\ref{fig:ncp-extreme}). In this
    case
\beas&\lambda_1=\lambda_2=\sqrt{1-\cos^2(u)\cos^2(v)}\,,\\
&\lambda_3=\sin(u)\sqrt{1-\cos^2(u)\cos^2(v)}\,,\\
&t=\sin(u)\sin^2(v)\,. \eeas For $u\ne 0\ne v$ the map is
definitely not a unitary one (its image is a proper subset of the
unit sphere) and in its image there are more than 3 pure states
(in fact the whole circle of states at which the ellipsoid touches
the unit sphere). From Theorem~\ref{theo:2} it follows thus that
the map cannot be a completely positive one. Indeed, for the
chosen values of $\lambda_1,\lambda_2,\lambda_3$, and $t$ the map
is an extreme positive \cite{gorini76}. The fact that it is not
completely positive can be checked independently by finding that
its image under the Jamio{\l}kowski isomorphism
\cite{jamiolkowski72} is not positive semi-definite.

\vskip 2cm

\begin{figure}[H]
\begin{center}
\includegraphics[scale=0.4]{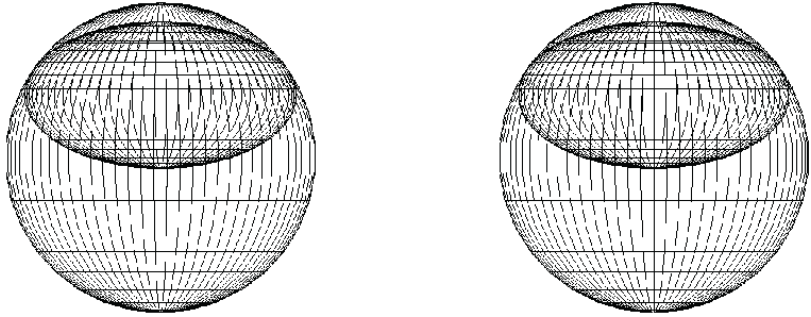}
\caption{An extreme positive map having a continuous family of
pure states in its image. Projections along two perpendicular
axes.} \label{fig:ncp-extreme}
\end{center}
\end{figure}

\end{enumerate}

\subsection{An extreme completely positive map having no pure states in its image}

Let us consider a CP map on $\mathbb{C}^{3\times 3}$ defined by
the following Kraus operators

\begin{equation}
V_1= \left[ \begin {array}{ccc} 1/\sqrt {3}&0&0\\
\noalign{\medskip}0&1/\sqrt {2}&0\\
\noalign{\medskip}0&0&{\frac {1}{\sqrt {1+{\alpha}^{2}}}}
\end {array} \right],
\quad
V_2=\left[ \begin {array}{ccc} 0&1/\sqrt {3}&0\\
\noalign{\medskip}0&0&1/\sqrt {2}\\
\noalign{\medskip}{\frac {\alpha}{\sqrt {1+{\alpha}^{2
}}}}&0&0\end {array} \right], \quad
V_3= \left[ \begin {array}{ccc} 0&0&1/\sqrt {3}\\
\noalign{\medskip}0&0&0\\
\noalign{\medskip}0&0&0\end {array} \right].
\end{equation}
A straightforward calculation gives
$V_1V_1^\dagger+V_2V_2^\dagger+V_3V_3^\dagger=I$.

For $\alpha=0$ the matrices $V_iV_j^\dagger$ form a basis in the
space of $3\times 3$ matrices, hence it is also true for small
$\alpha$. The map $A\rho=\sum_{i=1}^3 V_i^\dagger\rho V_i$ is thus
an extreme CP map (theorem \ref{choi1}). For $\alpha\ne 0$ there
is no $\ket{y}$ such that $V_i^\dagger\ket{x}\sim\ket{y}$ for some
$\ket{x}$ and $i=1,2,3$, hence $A$ does not send any pure state
into a pure one.

\section{Acknowledgements}
We thank G.~Esposito for reading the manuscript and making useful suggestions. M. K. acknowledges the support by the LFPPI Network financed by the
Polish Ministry of Science and Higher Education and the EU Program
SCALA (IST-2004-015714).

%\bibliographystyle{unsrt}
%\bibliography{wigner}

\end{document}